\documentclass{iopconfser}

\usepackage{amsmath}
\usepackage{amsthm}
\usepackage{amssymb}
\usepackage[square,sort&compress,numbers]{natbib}
\usepackage{upgreek}
\usepackage{wrapfig}
\usepackage{bm}
\usepackage[pdftex]{graphicx}
\usepackage{amsfonts}
\usepackage{mathrsfs}
\usepackage{color}
\usepackage{hyperref}

\newcommand{\bd}{\begin{definition}}                
\newcommand{\ed}{\end{definition}}                  
\newcommand{\bc}{\begin{corollary}}                 
\newcommand{\ec}{\end{corollary}}                   
\newcommand{\bl}{\begin{lemma}}                     
\newcommand{\el}{\end{lemma}}                       
\newcommand{\bp}{\begin{proposition}}            
\newcommand{\ep}{\end{proposition}}                
\newcommand{\bere}{\begin{remark}}                  
\newcommand{\ere}{\end{remark}}                     

\newcommand{\bt}{\begin{theorem}}
\newcommand{\et}{\end{theorem}}

\newcommand{\be}{\begin{equation}}
\newcommand{\ee}{\end{equation}}

\newcommand{\bit}{\begin{itemize}}
\newcommand{\eit}{\end{itemize}}
\newtheorem{theorem}{Theorem}[section]
\newtheorem{corollary}[theorem]{Corollary}
\newtheorem{lemma}[theorem]{Lemma}
\newtheorem{proposition}[theorem]{Proposition}
\theoremstyle{definition}
\newtheorem{definition}[theorem]{Definition}
\theoremstyle{remark}
\newtheorem{remark}[theorem]{Remark}

\newcommand{\dd}{{\rm d}}

\begin{document}

\title{Destructuring Physics: A functional derivation of spacetime\footnote{Contribution to the Proceedings of the Eleventh International Workshop  DICE 2024,  Spacetime - Matter - Quantum Mechanics "Quo vadis, fisica?", Castello Pasquini, Castiglioncello (Italy),  September 16--20,  2024. J. Phys.: Conf. Ser. {\bf 3017} (2025) 012004, DOI:10.1088/1742-6596/3017/1/012004.}
}

\author{Ettore Minguzzi}

\affil{Dipartimento di Matematica, Universit\`a degli Studi di Pisa,  Via
B. Pontecorvo 5,  I-56127 Pisa, Italy. }

\email{ettore.minguzzi@unipi.it} $\empty$

\begin{abstract}
I propose that Physics should be formulated using minimal mathematical structure, beginning with its foundational arena: spacetime. This paper opens with a concise overview of several research directions explored in previous work. Among these are the proposal to represent spacetime at the quantum scale using (measure) closed ordered spaces; the unification of causality and topology through quasi-uniformities; the concept of the product trick to unify causality and metricity; the introduction of upper semi-continuous (stable) Lorentzian distances; the representation of spacetime via steep time functions; and the formulation of the Lorentzian distance formula.

Subsequently, the properties of the stable distance over stable spacetimes are used as a reference to propose a simplified, abstract notion of spacetime. The discussion shows that spacetime can be introduced in a general and minimalistic way as nothing more than a family of functions defined over an arbitrary set. This abstraction removes unnecessary mathematical complexity, reducing spacetime to its  essential elements while preserving its most fundamental physical  properties.
\end{abstract}

\section{Some Recollections}

This work is essentially an update of a contribution I made to this conference twelve years ago, where I first presented some ideas on representing spacetime using measure closed ordered spaces \cite{minguzzi13e}. Back then, inspired by causal set theory \cite{dowker03,surya19}, I argued that spacetime should be described using a mathematical structure capable of incorporating elements of discreteness---something hinted at by quantum mechanics---rather than relying solely on the notion of a smooth manifold. After all, it's likely that the unification of gravity and quantum mechanics will reveal the smooth manifold as merely an emergent phenomenon at low energy scales.

Since the Lorentzian metric can be reconstructed from the conformal factor and the distribution of light cones \cite[App.\ D]{wald84}, I explored the possibility that spacetime could  be represented by a spacetime measure and a closed order encoding causality (we'll get to why \emph{closed} matters in a moment). At high energies or low regularity, the concept of a cone distribution might break down---there might not even be a tangent bundle---but causality, understood as a relation, would still be present.

In that earlier paper, I not only motivated this line of inquiry physically but also explored the idea of unifying the closed order and topology into a single concept called a \emph{quasi-uniformity}. This notion, first discovered by Nachbin \cite{nachbin65}, is surprisingly natural---it is essentially a generalization of the standard (Weil) topological uniformity, obtained by dropping a symmetry axiom.

But before diving deeper, let me take a step back and explain what initially drew me to closed relations and the related notion of upper semi-continuous Lorentzian distance.


\subsection{Causality in General Relativity}

The global differential geometric study of spacetime really took off in the mid-1960s and early 1970s, following the publication of the Raychaudhuri equation in the 1950s. This was the era of groundbreaking contributions by Hawking, Geroch, Penrose, and Tipler, culminating in the singularity theorems by Penrose and Hawking, and the publication of Hawking and Ellis's seminal book \cite{hawking73}. The subsequent decades, including the 1980s and 1990s, were more about systematization, culminating in the publication of influential books by Beem, Ehrlich and Easley \cite{beem96}, and by O'Neill \cite{oneill83}.

While causality was central to many of these results---ultimately earning Penrose a Nobel Prize---most researchers in relativity at the time focused on more algebraic or tangible topics. My own interest in mathematical relativity began during my PhD, when I studied Hawking and Ellis's book. But I only truly entered the field by chance, after realizing that some results by Erasmo Caponio and Antonio Masiello \cite{caponio04d} on the Lorentz force equation---which I heard about at a conference---could be improved using Kaluza-Klein theory, a topic I had dabbled in earlier. To make progress, I needed to revisit the chapter on causality in Hawking and Ellis's book, and that's when I fell down the rabbit hole of global methods in mathematical relativity \cite{caponio03}. This fascination deepened when Miguel Sánchez invited me to co-author a review of causality theory \cite{minguzzi06c}.


\subsection{The Chronological vs.\ Causal Relation}

In the 1970s and 1980s, most researchers working on causality focused on the \emph{chronological relation}, which had nice topological properties---it was open and, in strongly causal spacetimes, the manifold topology could be described using chronological diamonds $I(p,q) = I^+(p) \cap I^-(q)$ (i.e., the Alexandrov topology). In contrast, the \emph{causal relation} $J$ was less well-behaved, often failing to be closed unless strong causality conditions like global hyperbolicity were imposed. For example, simply removing a point from Minkowski spacetime illustrates this issue.

One might wonder why this matters---after all, isn't spacetime globally hyperbolic, as suggested by the strong cosmic censorship conjecture? The problem is that exact solutions in general relativity can exhibit all sorts of causal pathologies. Penrose, for instance, showed that plane waves provide a family of examples where such pathologies arise \cite{penrose65,hubeny05}.
The idea of cosmic censorship is that causal pathologies should not form starting from a well behaved spacetime, not that  they cannot be present in exact solutions of Einstein's equation.
Ultimately, we need a mathematical framework capable of handling such pathologies before we can rule them out in specific physical scenarios. Unsurprisingly, axiomatic approaches to causality---dating back to the early days of relativity---almost always focused on the chronological relation rather than the causal one (except in trivial cases like Minkowski spacetime). This is why strong causality, closely tied to the chronological relation, was one of the most commonly imposed conditions in the 1970s and 1980s.


\subsection{A Shift in Perspective}

Things began to change, however. I became interested in weaker causality conditions after noticing that the causal ladder in Hawking and Ellis's book could be expanded to include new levels. Conditions related to the non-total or non-partial imprisonment of causal curves in compact sets, which were often treated as disconnected from other causality conditions, could actually be integrated into the causal ladder. This line of inquiry required developing tools that weren't realized to their full potential at the time, particularly limit curve theorems. Hints to what should have been done came from papers by Galloway and others. I worked  then on a general version of the limit curve theorem and realized that many proofs in the literature, which assumed strong causality, could actually be proven under the weaker assumption of non-total imprisonment. This was because non-total imprisonment had specific, desirable consequences when combined with the limit curve theorem---a cornerstone of causality theory. This hinted that the chronological relation might not be as central as previously thought.

Another clue came from the study of time functions---continuous functions that increase strictly along causal curves. Seifert \cite{seifert71} had introduced a  reflexive, transitive (i.e.\ a preorder) closed relation $J_S$ by intersecting the causal relations $J_{g'}$ for metrics $g'$ with wider causal cones than $g$: $J_S := \cap_{g'>g} J_{g'}$. Seifert's work was ahead of its time but unfortunately contained gaps and incomplete arguments. Later researchers, including myself, had to reprove or improve some of his claims. In particular, I presented a new proof \cite{minguzzi07} (see also \cite{hawking74})  that the antisymmetry of $J_S$ (making it an order) is equivalent to the existence of a time function---a property known as stable causality. This led me to suspect a more general result: the existence of an increasing function for some order should follow directly from the order being closed. This search eventually led me to Auslander-Levin's theorem \cite{auslander64,levin83} and to a body of results which at the time was mostly confined in the literature in mathematical economics \cite{minguzzi09c}. Most importantly, I became aware of Nachbin's work on closed preordered spaces \cite{nachbin65}, which turned out to be exactly what I needed.


\subsection{The Relevance of Closed Relations}

Nachbin's work reveals that the topological theory we learn in university is just a special case of a more general theory involving a triple $(M, \mathscr{T}, \le)$, where $\mathscr{T}$ is a topology and $\le$ is a preorder whose graph $G(\le) = \{(p,q): p \le q\}$ is closed in the product topology. Standard topology corresponds to the case where $G$  is  the diagonal $\Delta \subset M \times M$---the discrete order. The closure of the diagonal is nothing but the Hausdorff condition. This makes it clear why working with a closed relation is preferable: \emph{the closure of the relation generalizes the Hausdorff condition}. The Hausdorff condition itself is recovered from the closure of the preorder provided it is an order as in that case\footnote{If $R$ is a relation $R^{-1}=\{(x,y): (y,x)\in R\}$.} $G\cap G^{-1}=\Delta$ is closed, an equivalent formulation of Hausdorffness. A fully general topological theory can be developed for closed relations, with preorder analogs of standard topological properties (e.g., normal preorder regularity, complete preorder regularity, quasi-pseudo-metrizability), where continuous functions are replaced by continuous isotone functions. This type of generalization isn't possible with open relations.


\subsection{Which Closed Relation?}

Before diving deeper into the implications, let us revisit the issue of open versus closed relations. While the chronological relation $I$ is the obvious candidate for an open relation, there can be many choices for a closed relation containing $I$. In an axiomatic approach to causality, one would ideally like to show that any smooth spacetime has a single, natural closed preorder describing its causality.

We have already mentioned Seifert’s closed relation $J_S \supset J \supset I$, which is constructed using sequences of $g'$-timelike connecting curves for $g' > g$. However, another approach is to consider the smallest closed, reflexive, and transitive relation that includes $I$.
In globally hyperbolic spacetimes (and causally simple ones), this relation reduces to the ordinary causal relation $J$. However, since we aim to avoid imposing overly strong causality conditions, the smallest closed preorder containing $I$---denoted $K$---becomes particularly significant. This relation was introduced by Sorkin and Woolgar \cite{sorkin96}, who anticipated its importance in low-regularity settings. Their work was framed in the context of manifolds endowed with a $C^0$ Lorentzian metric, where the usual smoothness assumptions no longer hold.
%
%
%
%
%

Before any causal axiomatic approach could be developed, it was crucial to show that, in the smooth case, there is essentially only one natural closed relation. This meant proving the equivalence of the $K$-relation with Seifert’s $J_S$-relation. I demonstrated this \cite{minguzzi08b,minguzzi09c} by showing that the antisymmetry of $J_S$ is equivalent to the antisymmetry of $K$, and that, in this case, $K = J_S$. Later, it was established that, under the same assumptions, the manifold topology could be recovered from $K$ in a purely order-theoretic fashion \cite{ebrahimi15,sorkin19}.

\subsection{Quasi-Uniformity and the Unification of Topology and Order}

Perhaps the most fascinating aspect of Nachbin’s theory of closed preordered spaces is the concept of {\em quasi-uniformity}. A quasi-uniformity is a filter $\mathscr{U}$ on the diagonal of $M \times M$ that satisfies the condition: for every $U \in \mathscr{U}$, there exists $V \in \mathscr{U}$ such that $V \circ V \subset U$. Unlike a uniformity, there’s no symmetry condition: if $U \in \mathscr{U}$ then $U^{-1} \in \mathscr{U}$. As a result, $G = \bigcap \mathscr{U}$ is a preorder on $M$,  while the symmetrized uniformity $\mathscr{U}^* = \{ U \cap U^{-1} : U \in \mathscr{U} \}$ induces a topology $\mathscr{T}(\mathscr{U}^*)$. In essence, every quasi-uniformity induces a topological preordered space $(X,\mathscr{T}(\mathscr{U}^*), \bigcap \mathscr{U})$, where the topology is Hausdorff if and only if the preorder is an order.

This suggests that quasi-uniformity unifies topology and order into a single structure. In the context of spacetime, this means that the topology of spacetime and the causal order might both emerge from a single object: the spacetime quasi-uniformity.

To support this claim, one must first prove that smooth, stably causal spacetimes are quasi-uni\-for\-mi\-za\-ble, that is, they admit a quasi-uniformity from which both the topology and the  order $K(=J_S)$ can be derived. Nachbin had already shown that quasi-uni\-for\-mi\-za\-bi\-li\-ty is equivalent to {\em complete preorder regularity} \cite[Thm.\ 9, p.\ 67]{nachbin65}. A closed preordered space $(M, \mathscr{T}, \le)$ is completely regularly preordered if the topology and preorder can be recovered from the set of continuous isotone functions (where {\em isotone} means $x \le y \Rightarrow f(x) \le f(y)$). In other words, \( \mathscr{T} \) is the initial topology of the family of continuous isotone functions, and $x \le y$ if and only if $f(x) \le f(y)$ for every continuous isotone function $f$.

\subsection{Challenges in Preorder Topology}

In standard topology (where the order is discrete, i.e., $\Delta$), we have the following hierarchy of properties:
\[
\textrm{normal} \Rightarrow \textrm{completely regular (Tychonoff)} \Rightarrow \textrm{regular},
\]
\[
\textrm{locally compact and Hausdorff} \Rightarrow \textrm{completely regular (Tychonoff)},
\]
\[
\textrm{second countable and regular} \Rightarrow \textrm{metrizable}.
\]
In the preorder version of the theory, we start with the closed relation condition (analogous to the Hausdorff property) and hope to climb the ladder of topological properties using local compactness or second countability—assumptions implicit in the definition of a manifold. However, none of the above implications hold in the preorder version of the theory.

I spent some years working on this problem and eventually made progress. The first step was to prove the following result: {\em Every locally compact, $\sigma$-compact, closed preordered space is a normally preordered space.} In a normally preordered space, the preorder is indeed represented by the set of continuous isotone functions, but this does not necessarily hold for the topology. The missing ingredient is {\em convexity}: the topology must be generated by open future sets and open past sets (i.e., they form a subbasis).

In the context of topological preordered spaces, the  {\em increasing hull} (future) of a subset $S$ is denoted $i(S)$, while the {\em decreasing hull} (past) is denoted $d(S)$. A set is {\em increasing} (future) if $i(S) = S$, and {\em decreasing} (past)  if $d(S) = S$. A subset $C$ is {\em convex} if $C = i(C) \cap d(C)$—equivalently, it is the intersection of an increasing and a decreasing set. Every completely preordered space is convex, and with convexity, normally preordered spaces descend to completely regularly preordered spaces.

Unfortunately, convexity is a global property and should not be confused with the weaker property of {\em weak convexity}---the topology admits a basis of open convex sets---or with the even weaker property of {\em local convexity}---the neighborhood system at any point is given by convex neighborhoods. While convexity implies weak convexity, which in turn implies local convexity, the converse implications do not necessarily hold. This is because not all convex neighborhoods are open, and not all open convex sets are the intersection of an open increasing set and an open decreasing set.

Ultimately, I found conditions that guarantee local convexity and then convexity. The first result states that if the topology is locally compact and $\sigma$-compact (as for manifolds), local convexity is promoted to convexity and to preorder normality and hence to quasi-uniformizability \cite[Cor.\ 2.14]{minguzzi12d}. In order to obtain the missing property of local convexity we can either assume that the preorder is $k$-preserving, namely that the convex hull of compact sets is compact  (a property which should be understood, if the preorder is an order, as the analog of global hyperbolicity for closed ordered spaces) \cite[Thm.\ 3.3]{minguzzi12d} or that the topology and preorder are determined from local information (compactly generated) \cite[Cor.\ 4.14]{minguzzi12d}. The last result shows that in the spacetime manifold case $(M,\mathscr{T}, K)$ is quasi-uniformizable, where $\mathscr{T}$ is the manifold topology and $K$ is the Sorkin-Woolgar relation. Under stable causality, this is the same as $(M,\mathscr{T}, J_S)$ which is a closed ordered space.  In the globally hyperbolic case one can get that $(M,\mathscr{T}, J)$ is quasi-uniformizable either using $K=J$ or using the mentioned result for $k$-preserving closed ordered spaces.

\subsection{Which Quasi-Uniformity?}

Having established that a locally compact $\sigma$-compact closed ordered space which is $k$-preserving (or compactly generated) admits a quasi-uniformity, the next question is whether there exists some canonical spacetime quasi-uniformity. It turns out that the answer is affirmative, but we need to add a further ingredient: the Lorentzian distance.

Recently, in joint work with Stefan Suhr and Aleksei Bykov \cite{minguzzi22,minguzzi24b} we introduced the notion of {\em Lorentzian metric space (without chronological boundary)} which is a set $X$ endowed with a function (Lorentzian distance) $d: X\times X\to [0,\infty)$, such that
\begin{itemize}
 \item[(i)] it satisfies the reverse triangle inequality for chronologically related triples of events $x\ll y \ll z$ (where $I=\{d>0\}$ is the chronological relation),
 \item[(ii)]
     the chronological diamonds\footnote{We denote $d_p:=d(p,\cdot)$, $d^p=d(\cdot, p)$.} $I(p,q):=\{r: d_p(r) d^q(r)>0\}$,   are relatively compact in a topology $\mathscr{T}$ for which $d$ is also continuous, and
 \item[(iii)]
     $d$ (weakly) distinguishes events: `$d_p=d_q$ and $d^p=d^q$' implies $p=q$,
     \item[(iv)] $X=I(X)$, namely every points is contained in a chronological diamond.
\end{itemize}
 The topology is actually unique, Hausdorff, locally compact, and it coincides with the initial topology of the functions $\{d_p, d^p , p\in X\}$. It is finer than the Alexandrov topology and not necessarily coincident with it (coincidence follows under additional assumptions such as $p\in \overline{I^\pm (p)}$ for every $p\in X$). The causal relation $J$ is derived from
 \begin{equation} \label{jrpo}
 J:=\{(p,q): d_p\ge d_q \ \textrm{and} \  d^p\le d^q \}.
 \end{equation}
Globally hyperbolic smooth spacetimes $(M,g)$ provide the main example of Lorentzian metric spaces, which can indeed be understood as the low regularity versions of globally hyperbolic spacetimes. The family of Lorentzian metric spaces is so general that it comprises causal sets (type of oriented graphs with weighted edges). Lorentzian (pre)length spaces are obtained by adding a further axiom that expresses the connectedness  of chronologically related events by suitably defined causal curves.

It turns out that in a Lorentzian metric space the triple $(X,\mathscr{T}, J)$ is a $k$-preserving closed ordered space which admits two canonical quasi-uniformities \cite{minguzzi24b}.
The {\em (non-fine)  quasi-uniformity} is generated by sets of the following form
\begin{equation} \label{copp}
\{(p,q): (d_q(z)-d_p(z))<a, \quad (d^p(z)-d^q(z))<b \},
\end{equation}
for $z\in X$, $a,b\in (0,+\infty]$.
The {\em fine quasi-uniformity} is generated by sets of the following form
\begin{equation} \label{cypp}
\{(p,q): \sup_K(d_q-d_p)<a, \quad \sup_K(d^p-d^q)<b \},
\end{equation}
for $K$ compact set and $a,b\in (0,+\infty]$ (thanks to (ii) and (iv) the compact set could be replaced by the finite union of chronological diamonds to make manifest that this quasi-uniformity follows directly from information in $d$).

If the Lorentzian metric space is {\em countably generated}, namely
\begin{itemize}
\item[(v)] there is a countable subset $\mathcal{S}\subset X$ such that for every $z\in X$ there are $x,y\in \mathcal{S}$ such that $x\ll z \ll y$,
\end{itemize}
then the topology is $\sigma$-compact, second countable and Polish, and the fine quasi-uniformity is quasi-metrizable. I shall not enter into this property here since, in any case, the quasi-metric, contrary to the quasi-uniformity, is not canonical.

The representability via a canonical quasi-uniformity holds for much more general objects than Lorentzian metric spaces (see also \cite[Thm.\ 9, p. 69]{nachbin65})

\begin{proposition}
Let $(X, \mathscr{T},\le, \mathcal{F})$ be a quadruple where $X$ is a set, $\mathscr{T}$ a topology, $\le$ a preorder, and $\mathcal{F}$ a family of real function on $X$. The compatibility conditions
\begin{center}
\begin{tabular}{r l}
topology-preorder: & the preorder $\le$ is $\mathscr{T}\times \mathscr{T}$-closed, \\
preorder-function:  & the order is represented by the functions: $p\le q$ iff $f(p) \le f(q)$ for all $f \in \mathcal{F}$, \\
function-topology: & $\mathscr{T}$ is the initial topology of $\mathcal{F}$,
\end{tabular}
\end{center}
hold if and only if  the topological preordered space $(X, \mathscr{T},\le)$ admits the quasi-uniformity  generated by sets of the form
\[
\{(p,q): f(p)-f(q)<b\},
\]
where  $f\in \mathcal{F}$, $b>0$.
\end{proposition}

\begin{proof}
To the right,  the intersection of all the elements of the quasi-uniformity is $\{(p,q): f(p)-f(q)\le 0\}$ which is the graph of $\le$ by the second condition. The symmetrized uniformity is generated by the sets of the form $\{(p,q): \vert f(p)-f(q)\vert <b\}$ which is the initial topology of $\mathcal{F}$ hence it coincides with $\mathscr{T}$ by the third condition.

To the left, the second and third properties are clear. The generated order $\le$ is closed became this is a standard property of the preorder induced by a quasi-uniformity.
\end{proof}

The following is a kind of specialization to $\mathcal{F}=\{d_z, -d^z: z\in X\}$.
\begin{proposition}
Let $(X, \mathscr{T},\le, d)$ be a quadruple where $X$ is a set, $\mathscr{T}$ a topology, $\le$ a preorder, and $d: X\times X\to [0,\infty)$ a function. Suppose that the following compatibility conditions hold
\begin{center}
\begin{tabular}{r l}
topology-preorder: & the preorder $\le$ is $\mathscr{T}\times \mathscr{T}$-closed, \\
preorder-function:  &The graph $J$ of $\le$ is expressed by formula (\ref{jrpo}), \\
function-topology: &$\mathscr{T}$ is the initial topology of the functions $\{d_z, d^z: z\in X\}$,
\end{tabular}
\end{center}
then the closed preordered space $(X, \mathscr{T},\le)$ admits the quasi-uniformity (\ref{copp}). Furthermore, if $d$ is $\mathscr{T}\times \mathscr{T}$-continuous then $(X, \mathscr{T},\le)$ admits the fine quasi-uniformity (\ref{cypp}).
\end{proposition}

Observe that $d$ need not satisfy a reverse triangle inequality, for this reason we do not use the word {\em Lorentzian distance} and we prefer to use the generic word {\em function}. Still $J$ defined by  (\ref{jrpo}) is a preorder.
\begin{proof}
Indeed, the intersection of all the elements of the quasi-uniformity is $\{(p,q):  d_q- d_p \le 0 \ \textrm{ and } \ d^p-d^q \le 0\}$ which is the expression for $J$. The symmetrization of the quasi-uniformity is the uniformity generated by
\begin{equation}
\{(p,q): \vert d_q(z)-d_p(z)\vert <a, \quad \vert d^p(z)-d^q(z)\vert<b \}=\{(p,q): \vert d^z(q)-d^z(p)\vert <a, \quad \vert d_z(p)-d_z(q)\vert<b \},
\end{equation}
for arbitrary $z\in X$, $a,b>0$, which is the initial topology of the functions  $\{d_p, d^p: p\in X\}$.

The proof of the last statement goes as the proof of \cite[Thm.\ 7.13]{minguzzi24b}.
\end{proof}

\subsection{Consequences of the Existence of a Quasi-Uniformity: Compactification and Completion}

A closed ordered space $(X,\mathscr{T}, \le)$ is completely regularly ordered iff it is quasi-uniformizable with Hausdorff topology \cite{nachbin65}. For these spaces there is a standard procedure due to Nachbin, which generalizes the Stone–Čech compactification of Tychonoff spaces, hence called {\em Nachbin compactification}. I suggested that spacetimes could be compactified using this procedure in   \cite{minguzzi12d,minguzzi13e,minguzzi18b} \cite[Sec.\ 2]{minguzzi11b}.

It consists in considering the compact space $[0,1]^{\mathcal{F}}$ where $\mathcal{F}$ is the family of functions $f: X\to [0,1]$ which are continuous and isotone. The space $X$ can be mapped into it via the evaluation map $e: X\to [0,1]^{\mathcal{F}}$ such that $\pi_f(e(x))=f(x)$. The closure of the image is the Nachbin compactification $n X$. Once endowed with the topology and order induced from $[0,1]^{\mathcal{F}}$ it becomes a compact closed ordered space. The set $X$ is dense in $n X$. This order compactification is characterized by the property that every bounded continuous isotone function can be extended to the compactification. An interesting feature of this compactification is that the order is extended while remaining closed, a desirable property which does not seem to be shared by other construction that have been proposed in general relativity. The compactification $nX$ is, in a well defined sense, the largest possible compactification in a closed ordered space \cite[p.\ 81]{fletcher82}. All the other compactifications in a closed ordered space are obtained from $nX$ via identification of points (quotient). Naturally, we can call  $n X\backslash X$ the {\em boundary} of spacetime. Still, it is somewhat large so one can consider other methods to attach a boundary to a spacetime.

Observe that we have shown not only that in most cases of interest the topological ordered space of spacetime is quasi-uniformizable, we have also shown that in many cases the quasi-uniformity is canonical and derived from the Lorentzian distance $d$. This opens the possibility of completing the spacetime via a canonical procedure of completion for the quasi-uniformity (there are two canonical quasi-uniformities so there will be corresponding completion procedures depending on which one is chosen). There are several approaches but the simplest seems to be  Fletcher and Lindgren's {\em bicompletion} \cite{fletcher82}.

In the remainder of this subsection we shall only consider quasi-uniformities that induce Hausdorff topologies.
Let us denote with $\mathscr{U}:=\mathscr{Q}^*$ the uniformity induced by the quasi-uniformity $\mathscr{Q}$. The quasi-uniformity is said to be {\em bicomplete} if $\mathscr{U}$ is complete.
If it is not, it can be bicompleted to a quasi-uniform space which is bicomplete \cite[Thm.\ 3.33]{fletcher82}. The interesting fact is that the bicompletion is unique and accomplishes not only  the extension of the uniformity (this would be a more or less standard result) but also the extension of the quasi-uniformity and hence of the order associated with it.

Let us summarize the construction of this completion. A filter $\mathscr{F}$ of subsets of $X$  is said to be a {\em $\mathscr{U}$-Cauchy filter} if for every $U\in \mathscr{U}$ we can find $F\in \mathscr{F}$ which is $U$-small, that is $F\times F\subset U$. For instance, any sequence $\{x_k\}$ defines a filter $F$ including the subsets of $X$ that contain tails of the sequence.  For the quasi-uniformity (\ref{copp}) the sequence is Cauchy if for every $z\in X$, and $a,b>0$, we can find a positive integer $N$ such that for $k, k'\ge N$,
\[
\vert d^z(x_{k'})-d^z(x_k)\vert <a, \quad  \vert d_z(x_k)-d_z(x_k')\vert <b.
\]
For the fine quasi-uniformity (\ref{cypp}) the point $z$ is replaced by a comapct set $K$, and the previous conditions by
\[
\sup_K \vert d_{x_{k'}}-d_{x_k}\vert <a, \quad  \sup_K \vert d^{x_k}-d^{x_k'}\vert <b.
\]
A convergent filter is a  $\mathscr{U}$-Cauchy filter, $\mathscr{F}$ being given by the sets including the point of convergence.  A {\em minimal}  $\mathscr{U}$-Cauchy filter is one that does not properly contain any other  $\mathscr{U}$-Cauchy filter. Every  $\mathscr{U}$-Cauchy filter contains a unique minimal $\mathscr{U}$-Cauchy filter which, furthermore admits a base of $\mathscr{T}(\mathscr{U})$-open sets.

A function $f: (X, \mathscr{Q})\to  (X', \mathscr{Q}')$ is quasi-uniformly continuous if for every $Q\in \mathscr{Q}$ there is $Q'\in \mathscr{Q}'$ such that $(f\times f)^{-1}(Q')\subset Q$. A quasi-unimorphism is a quasi-uniformly continuous bijection together with its inverse. A bicompletion of a quasi-uniform space $(X,\mathscr{Q})$ is a bicomplete quasi-uniform space $(Y,\mathscr{P})$ that has a $\mathscr{T}(\mathscr{P}*)$-dense subspace quasi-unimorphic to $(X,\mathscr{Q})$.

Fletcher and Lindgren's bicompletion $(\tilde X, \tilde{\mathscr{Q}})$ is constructed as the set of minimal $\mathscr{U}$-Cauchy filters by extending the quasi-uniformity to this set. The base for the extended quasi-uniformity is given by sets of the form
\[
\tilde Q:=\Big\{(\mathscr{F},\mathscr{G}): \exists F\in \mathscr{F}, G\in \mathscr{G}  \textrm{ such that } F\times G\subset Q \Big\},
\]
where $Q\in \mathscr{Q}$. As mentioned the bicompletion is essentially unique as any other bicompletion is quasi-unimorphic to Fletcher and Lindgren's bicompletion $(\tilde X, \tilde{\mathscr{Q}})$, see \cite[Thm.\ 3.34]{fletcher82}. In conclusion, under very mild conditions the spacetime can be attached a canonical boundary via (bi)completion. For space reasons we only sketched the theory. I postpone to future work a  more detailed analysis.

\subsection{The (stable) upper semi-continuous Lorentzian distance}
In order to deal with the metric aspects of spacetime, i.e.\ with proper time, it seems natural, in any  abstract approach to spacetime, to have at one's disposal a Lorentzian distance $d$.

The question is: what properties should be imposed on $d$ in an abstract setting? It seems worth investigating this question by first looking at what happens in a low regularity - weak causality theory that yet makes the assumption that the spacetime is a manifold. As a reference we can consider the theory of closed/proper Lorentz-Finsler spaces developed in \cite{minguzzi17}. In that framework the causal cones are provided by an upper semi-continuous distribution $x \mapsto C_x \subset T_xM\backslash 0$ of closed sharp non-empty convex cones, in other words a convex sharp closed subbundle of the slit tangent bundle (so defining a so called closed cone structure). The metric aspects are defined by a Finsler fundamental function $F: C_x \to [0,\infty)$ satisfying some properties which are best introduced by using the {\em product trick}. The idea is to study the spacetime $M^\times =M\times \mathbb{R}$, define on it a translationally invariant  cone structure by making use of function $F$, and then impose that such cone distribution is upper semi-continuous.
  At each point $P=(p,z)$ of $M^\times $ the cone $C_P^\downarrow$  to be considered is the subset of $T_pM\backslash 0\times \mathbb{R}$ given by the hypograph of the function (there is also a symmetrized version that we shall not recall here) \cite[Eq.\ (3.16)]{minguzzi17}
 \begin{equation} \label{curx}
 F^\downarrow=
 \begin{cases}
F(p),& \text{if } v \in C_p , \\
-\infty, & \text{if } v \notin C_p.
\end{cases}
 \end{equation}
 By definition, the triple $(M,C,F)$ is a {\em closed Lorentz-Finsler space} if $(M^\times, C^\downarrow)$ is a closed cone structure.  The notion of {\em proper Lorentz-Finsler space}  has an additional assumption that prevents the cones and function $F$ from collapsing in some directions. Within the latter more specialized structure it is possible to introduce the chronological relation, which, in general is not present.

For a closed Lorentz-Finsler space, or better for its closed cone structure, we can still define the property of stable causality as the stability of causality under enlargements of cones, we can still define the Seifert relation $J_S$, prove that its antisymmetry is equivalent to stable causality and to the existence of time functions and that in this case $K=J_S$ \cite[Thm.\ 3.16]{minguzzi17}. We can still define global hyperbolicity, with the caveat that the compactness of causal diamonds should be replaced by the preservation of compactness under the causally convex hull. In fact, we can recover the ladder  of causality properties, a result which suggests that the chronological relation could be less important than previously credited.

Also, on the metric side, it can be shown   that in globally hyperbolic spacetimes  the  Lorentzian distance is finite \cite[Prop.\ 2.26]{minguzzi17} and upper semi-continuous \cite[Thm.\ 2.60(g)(d)]{minguzzi17}. Now, to get the full continuity of $d$ it is necessary to work with more specialized objects, namely proper Lorentz-Finsler spaces and also assume stronger regularity properties (local Lipschitzness) \cite[Thm.\ 2.53]{minguzzi17}, see however \cite{ling24} for a solution in the $C^0$ setting. So, if we are under causality conditions weaker than global hyperbolicity, the properties of upper semi-continuity and finiteness are lost and that of lower semi-continuity requires some additional assumptions. It is quite unpleasant to work with non-continuous functions, which is why, in the development of the abstract notion of Lorentzian metric space \cite{minguzzi22,minguzzi24b} we imposed the continuity and finiteness conditions along with the compactness condition on diamonds: in that theory we are building low regularity versions of global hyperbolicity which is why we can expect continuity to hold.

Does it mean that everything is lost under low regularity-weak causality? Well, the good news is that there is a very natural Lorentzian distance in this regime, we denote it $D$ in this section, but it is not defined via the standard procedure. Curiously, it is upper semi-continuous, not lower semi-continuous and, unlike $d$, it does not require additional properties on the spacetime to get a semi-continuity property  (such as global hyperbolicity or local Lipschitzness and properness). Furthermore, it has wider applicability because it satisfies the reverse triangle inequality for a larger relation than the causal relation $J$, namely $J_S$. In fact, I would argue that $d$ has been given so much importance out of an historical accident. If $D$ had been discovered before $d$, probably less work would have been devoted to study the properties of $d$.
I defined the value $D(p,q)$ by following the idea that leads to the definition of the Seifert relation. We open slightly the cones and consider the causal curves, for the enlarged cones, connecting the two points. But now, as we are interested on both the causal and  metric aspects, we not only enlarge the cones, we also enlarge the {\em indicatrices} inside them, meaning that the convex side of the enlarged indicatrix contains the previous indicatrix (this is really an enlargement of the cone in the product spacetime according to the product trick). This allows us to compute the proper time for each enlargement and then to take the infimum which, by definition, is $D(p,q)$. We call it the {\em stable Lorentzian distance.} As mentioned, the so obtained function is upper semi-continuous in the product manifold topology, satisfies the reverse triangle inequality for triples $(p,q), (q,r) \in J_S$, satisfies $D(p,q)>0 \Rightarrow (p,q)\in J_S$, and for stably causal spacetimes, it satisfies $D(p,p)=0$. Moreover, $d\le D$ and under global hyperbolicity $d=D$.

An important concept associate to $D$ is that of {\em stable spacetime}. This is a property which stays between global hyperbolicity and stable causality \cite[Thm.\ 2.63]{minguzzi17} and obtained by imposing that $D$ is finite. A theorem \cite[Thm.\ 2.61]{minguzzi17} then tells us that stable spacetimes are stable, in the sense that the enlargement of cones and indicatrices previously mentioned can be chosen so that both causality and the finiteness of $D$ are preserved.

The importance of stable spacetimes can be readily understood mentioning two results.


The first concerns the problem of establishing whether the spacetime can be embedded as a Lorentzian submanifold of the canonical example of flat spacetime, namely Minkowski's.
In the non-Finslerian/iso\-tro\-pic case, under sufficient differentiability of the metric ($C^3$), the  stable spacetimes are characterized as being the Lorentzian submanifolds (possibly with boundary) of Minkowski spacetime $\mathbb{L}^n$ for  some dimension $n$  \cite[Thm.\ 4.13]{minguzzi17}.

The second problem  is that of representing the spacetime geometry through suitable sets of functions living in it.
It is solved identifying first the geometrical ingredients.
If one focuses on just the topological ordered space $(M,\mathscr{T},J_S)$, namely on a closed cone structure under stable causality, then in order to recover $J_S$ we need to prove the formula $(p,q)\in J_S$ iff for every $f$,  $f(p)\le f(q)$, where $f$ runs over the continuous isotone functions. This formula is a consequence of order normality which holds for closed orders over manifolds as it follows from Auslander-Levin's theorem and other results \cite{auslander64,levin83,minguzzi11c}.
In order to represent the topological ordered space one also wants to recover the topology which should be the initial topology of the family of continuous isotone functions. The property that needs to be proved is complete order regularity (or quasi-uniformizability). This property, which we already discussed, is harder to prove but still holds for the Seifert (or Sorkin and Woolgar) order induced on a stably causal closed cone structure \cite{minguzzi12d}. What we really want to recover is, the topology, the order and the Lorentzian metric from a suitable family of functions. The formula to be proved was conjectured by Parfionov and Zapatrin \cite{parfionov00} in the context of Connes' non-commutative approach to the unification of fundamental forces. They conjectured the {\em Lorentzian distance formula}
\[
d(p,q)=\inf\{ [f(q)-f(p)]^+ : f \in \mathscr{S} \}
\]
 where $\mathscr{S}$ is the family of $F$-steep temporal functions and where $a^+:=\textrm{max}\{0,a\}$. We recall that a function $f$ is {\em temporal} if it is $C^1$ and $\dd f$ is positive on the future causal cone $C_p$, for every $p$, and  {\em $F$-steep} if $\dd f(v)\ge F(v)$   for every $v\in C$.

Observe that it makes sense to consider differentiable functions because we are considering a smooth manifold. In a more abstract setting the rushing functions \cite{minguzzi18b} would need to be considered instead, we shall return on this point later on.

Now, we proved the formula for globally hyperbolic spacetimes \cite[Thm.\ 4.67]{minguzzi17}, but it is really more natural under a not so strong causality condition. The fact that in it we have the infimum implies that the Lorentzian distance will be upper semi-continuous, not necessarily lower semi-continuous. Fortunately, we had already noticed that the most natural Lorentzian distance is $D$ and hence upper semi-continuous.
Indeed, we were able to show that it is possible to recover $D$ in general, and the reason one recovers $d$ under global hyperbolicity is that under that assumption $D=d$ \cite[Thm.\ 2.60(g)]{minguzzi17} (see also \cite[Thm.\ 4.8]{minguzzi17}). We proved \cite[Thm.\ 4.6]{minguzzi17}

\begin{theorem} \label{aas}
Let $(M,{F})$ be a closed Lorentz-Finsler space and let $\mathscr{S}$ be the family of smooth   ${F}$-steep temporal functions. The Lorentz-Finsler space  $(M,{F})$ is stable if and only if $\mathscr{S}$  is non-empty. In this case $\mathscr{S}$   represents
\begin{itemize}
\item[(i)] the order $J_S$, namely $(p, q)\in J_S \Leftrightarrow f(p)\le f(q), \ \forall f \in \mathscr{S}$;
\item[(ii)] the manifold topology, as it is the initial topology of the functions in $\mathscr{S}$;
\item[(iii)] the stable distance, in the sense that  the distance formula holds true: for every $p,q\in M$
\begin{equation}
 D(p,q)=\mathrm{inf} \big\{[f(q)-f(p)]^+\colon \ f \in \mathscr{S}\big\}.
\end{equation}
\end{itemize}
\end{theorem}

\begin{remark}[Lorentz-Wasserstein distance under weak causality conditions] $\empty$\\
In the study of the causal optimal transport problem \cite{miller17,eckstein17,suhr16} for any two probability measures $\mu,\nu$ on $M$ one says that they are causal if there is a coupling $\omega$ supported on $J$. Then one considers the $s^\textrm{th}$ Lorentz-Wasserstein distance
\[
LW_s(\mu,\nu)=\sup_{\omega \in \Pi_c(\mu,\nu)}\left[\int d(p,q)^s \dd \omega(p,q)\right]^{1/s} ,
\]
where $\Pi_c(\mu,\nu)$ is the set of causal couplings. It represents the optimal total gain of moving the mass from $\mu$ to $\nu$ while respecting causality (the larger the average proper time the larger the gain/benefit; note that it is better not to speak of cost as the optimization would suggest a minimization). One can set the total gain to $-\infty$ if there is no causal coupling or to zero depending on the convention. Standard result on the existence of the optimal transport plan require, if one has to minimize the total cost, the lower semi-continuity of the cost function \cite[Thm.\ 2.10]{ambrosio21}. Since here we have to maximize, what we need is the upper semi-continuity of function $d$. Naturally, as $d$ is lower semi-continuous, this leads to study the problem of optimization of the causal transport of  measures in spacetimes in which $d$ is continuous, most notably globally hyperbolic ones.

However, as already noted by Miller \cite{miller17}, the transport of measures can be studied by replacing $J$ with the relation $K (=J_S)$ on stably causal spacetimes. This observation can be further refined by recognizing that the function to be used in the definition of the Lorentzian Wasserstein distance is the stable Lorentzian distance $D$, rather than $d$. By doing so, the upper semi-continuity property of $D$ aligns well with the requirements for the existence of an optimal transport plan, which can thus be guaranteed under significantly weaker causality conditions. This insight also extends to less regular frameworks, such as those where $M$ is not treated as a manifold. In any case, the function $d$ appearing in the displayed expression would need to be upper semi-continuous, like $D$, and the relation over which the causal couplings are supported would have to be closed, as is the case with $K$. In the transition to the smooth manifold case they would indeed correspond to $D$ and $K$, not to $d$ and $J$.

Of course, one might reach the conclusion that, from a technical standpoint, the theory of optimal transport works more effectively with an upper semi-continuous Lorentzian distance. However, without a clear physical motivation—such as the one we provided through the construction of the function $D$—it would be unlikely to consider these types of Lorentzian distances as a natural choice.

%
%
\end{remark}

\section{A Definition of Spacetime}


We now give a definition of spacetime that abstracts the properties of the stable distance in closed Lorentz-Finsler spaces \cite{minguzzi17} (see Theorem 2.6 of \cite{minguzzi17}, properties (a), (c), (d); under stable causality (h) is added besides the antisymmetry of the relation;   for stability see  \cite[Def.\ 2.29]{minguzzi17}). The correspondence with the notation of that paper is $X \leftrightarrow M$, $d\leftrightarrow D$, $\le \leftrightarrow J_S$.

\begin{remark}
For space reasons, in this exposition we shall omit most proofs. They will soon appear in a companion, more technical, work.
\end{remark}

For us the spacetime will be the following object
\begin{definition}
A {\em spacetime} is a quadruple $(X,\mathscr{T},\le, d)$, where  $(X,\mathscr{T},\le)$ is a  closed preordered space, and $d:X\times X\to [0,\infty]$ is an $\mathscr{T}\times \mathscr{T}$-upper semi-continuous function such that, $x\nleq y$ implies $d(x,y)=0$, and for every triple $x\le y\le z$, we have (reverse triangle inequality)
\[
d(x,y)+d(y,z)\le d(x,z).
\]
We say that it is {\em weakly stably causal}\footnote{One can be tempted to call this property just {\em causality} as there are no other causal relations, still I prefer to keep, at least in this work, the terminology that best clarifies the transition to the manifold case.} if $\le$ is an order and $d(x,x)=0$ for every $x\in X$. If, additionally, $d$ is finite, we say that it is {\em weakly stable}. We also call $(X,\mathscr{T},\le)$ the {\em causal structure} of the spacetime. It is {\em weakly stably causal} if $\le$ is an order. We drop {\em weakly} in the previous instances iff $(X,\mathscr{T},\le)$ is locally convex.
\end{definition}
Note that for every $x\in X$, $d(x,x)=0$ or $d(x,x)=+\infty$. The definition of (weak) stable causality for the causal structure is obtained from that of the spacetime dropping the condition on $d$, as it does not enter the causal structure.
A function $d$ that satisfies the above properties is called a {\em Lorentzian distance}.
For simplicity, we might denote the graph $G(\le)\subset X\times X$ with $\le$ itself.

The following reformulation will not be used but is worth mentioning.
\begin{definition}
A {\em spacetime}  is a triple $(X,\mathscr{T},\tau)$, where   $\tau:X\times X\to \{-\infty\}\cup [0,\infty]$ is an $\mathscr{T}\times \mathscr{T}$-upper semi-continuous function such that for every $x\in X$, $\tau(x,x)\ge 0$, and  for every $x,y,z\in X$
\[
\tau(x,y)+\tau(y,z)\le \tau(x,z),
\]
with the convention $-\infty+\infty=-\infty$.
We say that it is {\em weakly stably causal} if $\textrm{min}\{\tau(x,y),\tau(y,x)\} \ge 0 \Rightarrow x=y$ and for every $x\in X$, $\tau(x,x)=0$. If, additionally, $\tau<+\infty$ , we say that it is {\em weakly stable}.
\end{definition}

The two notions of spacetime are equivalent. To pass from the former to the latter set
 \begin{equation}
 \tau(x,y):=d^\downarrow(x,y):=
 \begin{cases}
d(x,y),& \text{if } x\le y, \\
-\infty,& \text{if } x\nleq y.
\end{cases}
 \end{equation}
To pass from the latter to the former set $\le=\{\tau\ge 0\}$, $d(x,y)=\textrm{max}\{0,\tau(x,y)\}$.

The function $\tau$ might be called {\em time separation} to distinguish it from $d$. As the notation $d^\downarrow$ suggests, it is really the metric analog of $F^\downarrow$, see Eq.\ (\ref{curx}), where $d$ is the metric analog of $F$. Function $\tau$ is convenient in the study of optimal transport on spacetime \cite{mccann18}.

\begin{remark}
A weakly stable spacetime admits the following simple characterization.
A {\em weakly stable spacetime}  is a triple $(X,\mathscr{T},\tau)$, where   $\tau:X\times X\to \{-\infty\}\cup [0,\infty)$ is a $\mathscr{T}\times \mathscr{T}$-upper semi-continuous function such that
\begin{align*}
\forall x,y\in X& &0\le \tau(x,y)+\tau(y,x)&\Leftrightarrow x=y, \\
\forall x,y,z\in X& &\tau(x,y)+\tau(y,z)&\le \tau(x,z).
\end{align*}
\end{remark}

The reason to include local convexity in  the definition of stable causality/stability will not be entirely appreciated in this work. It enters the proofs of results that will be presented in the following sections. It can be observed that in the smooth setting, under stable causality, the relation $K$ (which coincides with $J_S$) satisfies it \cite[Lemma 16]{sorkin96} \cite[Lemma 5.5]{minguzzi07}\cite[Thm.\ 4.15]{minguzzi12d}. Indeed, they are examples of compactly generated preorders \cite[Cor.\ 4.12]{minguzzi12d}.

\subsection{The Product Trick}
Let us arrive at the product trick which was introduced in \cite{minguzzi17,minguzzi17d} at the level of the tangent bundle. It shows that the metrical aspects of spacetime can be reduced to causality in a space with one additional dimension.

Given a set $X$ let us denote $X^\times:=X\times \mathbb{R}$. We are interested in a topological ordered space $(X^\times,\mathscr{T}^\times,$ $\le^\downarrow)$ structured so as to respect the product structure of $X^\times$.
So let $\mathscr{T}^\times$ be a topology on $X^\times$ which is the product of a topology $\mathscr{T}$ on $X$ by the topology $\mathscr{T}_\mathbb{R}$ of $\mathbb{R}$, and let $\le^\downarrow$ be a preorder on $X^\times$ that
satisfies
\[
(*) \qquad \qquad \textrm{ If } \ (p,r)\le^\downarrow (p',r') \ \textrm{ and } \ [s'-s]^+\le [r'-r]^+ \ \textrm{ then } \ (p,s)\le^\downarrow (p',s').
\]
This condition can be better understood as follows. Firstly, observe that a preorder that is translationally invariant,  namely for all $c\in \mathbb{R}$,  $(x,a)\le^\downarrow (y,b) \Rightarrow (x,a+c)\le^\downarrow (y,b+c)$, and can be projected to a preorder on $X$ defined by ``$x\le y$ iff there are $a,b\in \mathbb{R}$ such that $(x,a)\le (y,b)$'' (the translational invariance is used to show that this relation is transitive). Secondly, given a translationally invariant preorder $\le^\downarrow$ on $X^\times$, the preorder $\le^\downarrow$ contains the product of the projected preorder with the reverse canonical order\footnote{Here we are endowing the real line with the standard topology but with the reverse  canonical order. The down arrow recalls this fact. We could have used the standard order but we wanted to keep notations analogous to those of \cite{minguzzi17} which have some advantage when treating time functions as graphs, as some minus signs are not needed.} on $\mathbb{R}$: if $b\le a$ and $x\le y$ $\Rightarrow  (x,a)\le^\downarrow (y,b)$. The condition (*) is equivalent to these two conditions: translational invariance, which ensures projectability, and inclusion of the product preorder.

Note that by (*),  $(p,r)\le^\downarrow (p',r') \Rightarrow (p,0)\le^\downarrow (p',0)$ thus the projected order is ``$x\le y$ \textrm{ iff }  $(x,0)\le (y,0)$''.


\begin{definition} \label{pro}
A spacetime is a closed preordered space $(X^\times, \mathscr{T}^\times,\le^\downarrow)$ such that $\le^\downarrow$ satisfies (*). We say that it is {\em weakly stably causal} if $\le^\downarrow$ is  an order. If, additionally, the future (equiv.\ past) of every point does not contain an entire $\mathbb{R}$-fiber, we say that it is {\em weakly stable}. The adjective {\em weak} is dropped if  $(X^\times, \mathscr{T}^\times,\le^\downarrow)$ is locally convex.
\end{definition}

This notion is equivalent to the previous ones. Starting from $(X,\mathscr{T},\le, d)$ just let $X^\times =X\times \mathbb{R}$, $\mathscr{T}^\times =\mathscr{T}\times \mathscr{T}_\mathbb{R}$, and
\[
(x,a)\le^\downarrow (y,b) \ \ \textrm{ iff } \ \ x\le y  \ \textrm{and} \  b\le a+d(x,y) ,
\]
equivalently, starting from $(X,\mathscr{T},\tau)$ set
\[
(x,a)\le^\downarrow (y,b) \ \ \textrm{ iff } \ \ b\le a+\tau(x,y).
\]

On the other direction, starting from $(X^\times, \mathscr{T}^\times,\le^\downarrow)$ define $\le$ as the projected preorder
\[
x\le y   \ \textrm{ iff } \   (x,0)\le^\downarrow (y,0)
\]
and set $d(x,y):=0$ if $x\nleq y$ and   otherwise
\[
d(x,y):=  \sup \{b: (x,0)\le^\downarrow (y,b) \}
\]
 (note that  $b=0$  belongs to the set thus $d\ge 0$). Equivalently, these two conditions read
\[
\tau(x,y):= \sup \{b: (x,0)\le^\downarrow (y,b) \},
\]
where it is understood that $\tau=-\infty$ for the empty set.

\subsection{Spacetimes from Functional Spaces}

The next result proves that stable spacetimes are quite natural objects as they  follow from just a family of real functions over a set. It is inspired by \cite[Thm.\ 4.6]{minguzzi17} which was restricted to manifolds. Note that here the $F$-steep functions are replaced by the rushing functions \cite[Def.\ 1.26]{minguzzi18b}. On a spacetime $(X,\mathscr{T},\le, d)$ a function $X\to \mathbb{R}$ is {\em rushing} if
\[
\forall x,y \in X, \qquad x\le y  \ \Rightarrow \ f(x)+d(x,y) \le f(y)
\]
 equivalently, for every $x,y\in X$, $f(x)+\tau(x,y) \le f(y)$. Intuitively, a rushing function  interpreted as a time runs faster than proper time.
For the relationship between steep and rushing functions on a manifold, see \cite[Thm.\ 1.28]{minguzzi18b}.

\begin{theorem} \label{ckkr}
Let $\mathcal{F}$ be a family of real valued functions on $X$ that distinguishes/separates points: $f(x)=f(y)$ for every $f\in \mathcal{F}$ implies $x=y$. Then
\begin{itemize}
\item[(i)] The initial topology $\mathscr{T}$ generated by $\mathcal{F}$ is Tychonoff (hence Hausdorff). If $\mathcal{F}$ is countable then the initial topology is second countable, hence metrizable.
\item[(ii)] The  relation defined by
\begin{equation} \label{ord}
\le:=\{(x,y): f(x)\le f(y), \ \forall f\in \mathcal{F}\}
\end{equation}
 is an order which is closed with respect to the product of the initial topology of point (i).
\item[(iii)]    The function $d: X^2\to [0,\infty)$ defined by
\begin{equation} \label{dis}
d(x,y):= \max\big\{0, \inf_{\mathcal{F}} [f(y)-f(x)]\big\}
\end{equation}
is upper semi-continuous (with respect to the product of the initial topology of point (i)); satisfies $x\nleq y \Rightarrow d(x,y)=0$; for every $x\in X$, $d(x,x)=0$; and it  satisfies the reverse triangle inequality with respect to the relation $\le$ of point (ii).
\item[(iv)]
 Suppose that for every $\lambda \ge 1$, $\lambda \mathcal{F}\subset \mathcal{F}$, then we have the direct expressions
\begin{equation} \label{dis2}
\tau(x,y)= \inf_{\mathcal{F}} [f(y)-f(x)],
\end{equation}
and
\begin{equation}
(x,a)\le^\downarrow (y,b) \ \ \textrm{ iff } \ \ \forall f\in \mathcal{F}, \ \ f(x)-a\le f(y)-b.
\end{equation}
\end{itemize}
As a consequence, $(X,\mathscr{T},\le, d)$ is a stable spacetime and every element $f\in \mathcal{F}$ is $\mathscr{T}$-continuous, $\le$-isotone and $d$-rushing.
\end{theorem}

Without the condition on the distinction of points, the topology would not be $T_0$ and  $\le$ would not be antisymmetric (it would be a preorder). In principle one could work with indistinguishable points but then the Alexandrov quotient would need to be taken.

The initial topology induced by a family of functions $\mathcal{F}$ is denoted $\mathscr{T}_{\mathcal{F}}$, a preorder given by formula (\ref{ord}) is  denoted $\le_{\mathcal{F}}$, and a Lorentzian distance given by formula (\ref{dis}) is  denoted $d_{\mathcal{F}}$. Thus for any family of functions ${\mathcal{F}}$  that separates points, we have  a stable spacetime $(X,\mathscr{T}_{\mathcal{F}},\le_{\mathcal{F}}, d_{\mathcal{F}})$.

\begin{proposition} \label{mqpg}
On a spacetime $(X,\mathscr{T},\le, d)$ the $\mathscr{T}$-continuous, $\le$-isotone functions form a convex cone $\mathcal{I}$. The
$\mathscr{T}$-continuous, $\le$-isotone and $d$-rushing functions form a convex set $\mathcal{F}\subset \mathcal{I}$ such that $\lambda \mathcal{F}\subset \mathcal{F}$ for every $\lambda \ge 1$. Additionally, $\mathcal{I} + \mathcal{F} \subset \mathcal{F}$. For every $f,g\in \mathcal{F}$, $\textrm{min}(f,g), \textrm{max}(f,g)\in \mathcal{F}$, and a similar statement holds for $\mathcal{F}$ replaced by $\mathcal{I}$.
\end{proposition}

The set $\mathcal{I}$ is the analog of the causal cone on the tangent space of a Lorentzian manifold. The set $\mathcal{F}$ is the analog of the unit Lorentzian ball on the same tangent space. The nice fact is that the geometry of the tangent space is reproduced at the functional level, where the functional space aims to reproduce the whole spacetime.

\subsection{Representation Theorems}
%
%
%


We have established that a separating family of functions over a sets generates a stable spacetime. The natural question is whether there exists a kind of converse. We can indeed obtain that  every  spacetime, under suitable topological assumptions, is obtained in this way. This problem will be addressed in this section which is devoted to representation theorems.

Let us denote $x\sim y$ if $x\le y$ and $y\le x$. Let us also denote $x<y$  if $x\le y$ and $y\nleq x$.
A  isotone function $f: X\to \mathbb{R}$ such that $x<y \Rightarrow f(x)<f(y)$ is a {\em utility}.  If $\le$ is a partial order then a utility is defined by ``$x\le y$ and $x\ne y \Rightarrow f(x)<f(y)$''. If there is a function with the last property then $\le$ is necessarily a partial order. We call a continuous function that satisfies  ``$x\le y$ and $x\ne y \Rightarrow  f(x)<f(y)$'' {\em time function} (we avoid {\em strictly isotone} since some authors denote with this term a utility).

Our terminology is consistent with the result for the smooth case, since when $K$ is antisymmetric (stable causality) its continuous utilities are precisely the time functions \cite[Thm.\ 6]{minguzzi09c}, and if there is a time function $K$ is antisymmetric.

\begin{theorem} \label{cmgp}
A second-countable locally compact   closed preordered space (causal structure) $(X,\mathscr{T},\le)$ is weakly stably causal iff it admits a time function.
\end{theorem}


We say that $f:X\to \mathbb{R}$ is a rushing time function if it is both a rushing function and a time function.

The next result is the metric analog of the first statement (non-strict version) in \cite[Thm.\ 4.6]{minguzzi17}.
\begin{theorem}
A second-countable locally compact  spacetime $(X,\mathscr{T},\le,d)$ is weakly stable iff it admits a rushing time function.
\end{theorem}

%
%

The main result of this section is the following converse of Theorem \ref{ckkr}

\begin{theorem} \label{mai}
On a
locally compact $\sigma$-compact stable spacetime $(X,\mathscr{T},\le,d)$, denoting with $\mathcal{R}$ the family of continuous rushing functions, we have $\mathscr{T}=\mathscr{T}_{\mathcal{R}}$, $\le=\le_{\mathcal{R}}$ and $d=d_{\mathcal{R}}$. Under second-countability we can replace the continuous rushing functions with the rushing time functions.
\end{theorem}


The theorem applies also to   closed ordered spaces, just set  $d=0$ so that $\mathcal{R}=\mathcal{I}$, i.e.\ the family of continuous isotone functions. This gives:

\begin{theorem} \label{mai2}
On a
locally compact $\sigma$-compact stably causal closed ordered space (causal structure) $(X,\mathscr{T},\le)$,  we have $\mathscr{T}=\mathscr{T}_{\mathcal{I}}$, $\le=\le_{\mathcal{I}}$. Under second-countability we can replace the continuous isotone functions with the  time functions.
\end{theorem}

We believe the representation problems we just discussed will prove important in the context of Connes' non-commutative program for the unification of fundamental forces \cite{moretti03,besnard09,franco13,franco14b,besnard15}.

\section{Conclusions and Future Outlook}
The reader should now have formed an idea of where we are heading. We have shown that every (locally compact $\sigma$-compact) stable  spacetime is generated by the family of continuous rushing functions, in the sense that topology, order, and Lorentzian distance are deduced from that family. Thus, we can conclude that the spacetime can be ultimately regarded as a family of functions over a set. Moreover, every family of functions separating points induces a stable spacetime.



Suppose we are given a (pre)ordered vector space with some additional structure, including a convex set inside the positive cone. The convex set is invariant under multiplication by constants larger than one. Under which additional conditions is the vector space isomorphic to a suitable space of continuous functions over a spacetime? The positive cone would correspond to the time functions, and the convex set to the  rushing times.



In order to approach the problem we had first to clarify the nature  of the spacetime that we hope to recover and to motivate its definition. This is precisely what we accomplished in this work.
The second part of the problem will be addressed in a forthcoming work.



\begin{thebibliography}{10}
\expandafter\ifx\csname url\endcsname\relax
  \def\url#1{{\tt #1}}\fi
\expandafter\ifx\csname urlprefix\endcsname\relax\def\urlprefix{URL }\fi
\providecommand{\eprint}[2][]{\url{#2}}

\bibitem{minguzzi13e}
Minguzzi E 2013 {\em J. of Phys.: Conf. Ser.\/} {\bf 442} 012034 {C}ontribution
  to the proceedings of the conference 'DICE12, Space-Time-Matter-Quantum
  Mechanics, from the Planck scale to emergent phenomena', {C}astello
  {P}asquini, {C}astiglioncello (Italy) September 17 - 21, 2012

\bibitem{dowker03}
Dowker F, Henson J and Sorkin R~D 2003 {\em Mod. Phys. Lett.\/} {\bf A19}
  1829--1840

\bibitem{surya19}
Surya S 2019 {\em Living Reviews in Relativity\/} {\bf 22} 5

\bibitem{wald84}
Wald R~M 1984 {\em General Relativity\/} (Chicago: The {U}niversity of
  {C}hicago {P}ress)

\bibitem{nachbin65}
Nachbin L 1965 {\em Topology and order\/} (Princeton: D.\ {V}an {N}ostrand
  {C}ompany, {I}nc.)

\bibitem{hawking73}
Hawking S~W and Ellis G~F~R 1973 {\em The Large Scale Structure of
  Space-Time\/} (Cambridge: Cambridge {U}niversity {P}ress)

\bibitem{beem96}
Beem J~K, Ehrlich P~E and Easley K~L 1996 {\em Global Lorentzian Geometry\/}
  (New York: Marcel {D}ekker {I}nc.)

\bibitem{oneill83}
{O'N}eill B 1983 {\em Semi-{R}iemannian Geometry\/} (San Diego: Academic
  {P}ress)

\bibitem{caponio04d}
Caponio E and Masiello A 2004 {\em J. Math. Phys.\/} {\bf 45} 4134--4140

\bibitem{caponio03}
Caponio E and Minguzzi E 2004 {\em J. Geom. Phys.\/} {\bf 49} 176--186
  {arXiv}:gr-qc/0211100

\bibitem{minguzzi06c}
Minguzzi E and S\'anchez M 2008 {\em The causal hierarchy of spacetimes\/}
  ({\em {ESI} Lect. Math. {P}hys.\/} vol H. Baum, D. Alekseevsky (eds.), Recent
  developments in pseudo-Riemannian geometry) (Zurich: Eur. Math. Soc. Publ.
  House) pp 299--358 ISBN 978-3-03719-051-7 arXiv:gr-qc/0609119

\bibitem{penrose65}
Penrose R 1965 {\em Rev. {M}od. {P}hys.\/} {\bf 37} 215--220

\bibitem{hubeny05}
Hubeny V~E, Rangamani M and Ross S~F 2005 {\em Int. {J}. {M}od. {P}hys.\/} {\bf
  D14} 2227--2232

\bibitem{seifert71}
Seifert H 1971 {\em Gen. Relativ. Gravit.\/} {\bf 1} 247--259

\bibitem{minguzzi07}
Minguzzi E 2008 {\em Class. Quantum Grav.\/} {\bf 25} 015010
  {arXiv}:gr-qc/0703128

\bibitem{hawking74}
Hawking S~W and Sachs R~K 1974 {\em Commun. Math. Phys.\/} {\bf 35} 287--296

\bibitem{auslander64}
Auslander J 1964 {\em Contr. {D}ifferential {E}quations\/} {\bf 3} 65--74

\bibitem{levin83}
Levin V~L 1983 {\em Soviet Math. Dokl.\/} {\bf 28} 715--718

\bibitem{minguzzi09c}
Minguzzi E 2010 {\em Commun. Math. Phys.\/} {\bf 298} 855--868
  {arXiv}:0909.0890

\bibitem{sorkin96}
Sorkin R~D and Woolgar E 1996 {\em Class. Quantum Grav.\/} {\bf 13} 1971--1993

\bibitem{minguzzi08b}
Minguzzi E 2009 {\em Commun. Math. Phys.\/} {\bf 290} 239--248
  {arXiv}:0809.1214

\bibitem{ebrahimi15}
Ebrahimi N 2015 {\em J. Dyn. Sys. Geom. Theor.\/} {\bf 13} 1--41

\bibitem{sorkin19}
Sorkin R~D, Yazdi Y~K and Zwane N 2019 {\em Class. Quantum Grav.\/} {\bf 36}
  095006

\bibitem{minguzzi12d}
Minguzzi E 2013 {\em Topol. Appl.\/} {\bf 160} 965--978 {arXiv}:1212.3776

\bibitem{minguzzi22}
Minguzzi E and Suhr S 2024 {\em Lett. Math. Phys.\/} {\bf 114} 73
  {arXiv:}2209.14384

\bibitem{minguzzi24b}
Bykov A, Minguzzi E and Suhr S 2024 Lorentzian metric spaces and
  {GH}-convergence: the unbounded case arXiv:2412.04311

\bibitem{minguzzi18b}
Minguzzi E 2019 {\em Living Rev. Relativ.\/} {\bf 22} 3

\bibitem{minguzzi11b}
Minguzzi E 2012 {\em Appl. {G}en. Topol.\/} {\bf 13} 207--223 {arXiv}:1209.1839

\bibitem{fletcher82}
Fletcher P and Lindgren W 1982 {\em Quasi-uniform spaces\/} ({\em Lect. {N}otes
  in {P}ure and {A}ppl. {M}ath.\/} vol~77) (New York: Marcel {D}ekker, {I}nc.)

\bibitem{minguzzi17}
Minguzzi E 2019 {\em Rev. Math. Phys.\/} {\bf 31} 1930001 {arXiv}:1709.06494

\bibitem{ling24}
Ling E 2024 {\em Ann. Henri Poincar{\'e}\/}

\bibitem{minguzzi11c}
Minguzzi E 2012 {\em Appl. {G}en. Topol.\/} {\bf 13} 81--89 {arXiv}:1108.5123

\bibitem{parfionov00}
Parfionov G~N and Zapatrin R~R 2000 {\em J. Math. Phys.\/} {\bf 41} 7122--7128

\bibitem{miller17}
Miller T 2017 {\em Universe\/} {\bf 3} 27 proceedings of the conference Varying
  Constants and Fundamental Cosmology - VARCOSMOFUN'16, 12-17 September 2016,
  Szczecin, Poland. {arXiv}:1702.00702

\bibitem{eckstein17}
Eckstein M and Miller T 2017 {\em Ann. {H.} {P}oincar{\'e}\/} {\bf 18}
  3049--3096 {arXiv:}1510.06386

\bibitem{suhr16}
Suhr S 2018 Theory of optimal transport for {L}orentzian cost functions
  {arXiv:}1601.04532

\bibitem{ambrosio21}
Ambrosio L, Bru\'e E and Semola D 2021 {\em Lectures on optimal transport\/}
  (Cham, Switzerland: Springer)

\bibitem{mccann18}
McCann R~J 2020 {\em Camb. J. Math.\/} {\bf 8} 609--681 {arXiv:}1808.01536

\bibitem{minguzzi17d}
Minguzzi E 2018 {\em J. of Phys.: Conf. Ser.\/} {\bf 968} 012009 {C}ontribution
  to the proceedings of the conference 'Non-Regular Spacetime Geometry',
  {F}irenze (Italy) June 20 - 22, 2017

\bibitem{moretti03}
Moretti V 2003 {\em Rev. Math. Phys.\/} {\bf 15} 1171--1217

\bibitem{besnard09}
Besnard F 2009 {\em J. Geom. Phys.\/} {\bf 59} 861--875

\bibitem{franco13}
Franco N and Eckstein M 2013 {\em Class. Quantum Grav.\/} {\bf 30} 135007, 18
  pp

\bibitem{franco14b}
Franco N and Eckstein M 2014 {\em SIGMA Symmetry Integrability Geom. Methods
  Appl.\/} {\bf 10} Paper 010, 23 pp

\bibitem{besnard15}
Besnard F 2015 {\em Class. Quantum Grav.\/} {\bf 32} 135024, 26

\end{thebibliography}

\end{document}